\documentclass{article}

\usepackage{microtype}
\usepackage{graphicx}
\usepackage{subcaption}
\usepackage{placeins}

\usepackage{booktabs}

\usepackage{amsmath}
\usepackage{amssymb}
\usepackage{mathtools}
\usepackage{amsthm}
\usepackage{mathtools}
\usepackage{bbm}

\usepackage{comment}
\usepackage[textsize=tiny]{todonotes}
\usepackage{tikz}
\usepackage{hyperref}
\usepackage[capitalize,noabbrev]{cleveref}

\newcommand{\Ncal}{\mathcal{N}}
\newcommand{\Acal}{\mathcal{A}}

\DeclareMathOperator*{\expec}{\mathbb{E}}
\DeclarePairedDelimiter{\abs}{\lvert}{\rvert}

\theoremstyle{plain}
\newtheorem{theorem}{Theorem}[section]

\newtheorem{corollary}[theorem]{Corollary}
\theoremstyle{definition}
\newtheorem{definition}[theorem]{Definition}

\theoremstyle{remark}

\usepackage{wasysym}

  \usepackage[preprint]{neurips_2026}

\usepackage[utf8]{inputenc} 
\usepackage[T1]{fontenc}    
\usepackage{hyperref}       
\usepackage{url}            
\usepackage{booktabs}       
\usepackage{amsfonts}       
\usepackage{nicefrac}       
\usepackage{microtype}      
\usepackage[dvipsnames]{xcolor}         

\title{Simultaneous Gradient Learning \\ in First-Price Auctions}

\author{%
  Janik Bürgermeister\thanks{Corresponding author; \texttt{janik.buergermeister@tum.de}} \\
  Department of Computer Science\\
  Technical University of Munich\\
  Munich, Germany \\
  \And
  Julius Durmann \\
  Department of Computer Science\\
  Technical University of Munich\\
  Munich, Germany \\
  \And
  Martin Bichler \\
  Department of Computer Science\\
  Technical University of Munich\\
  Munich, Germany \\
  \And
  Mete Şeref Ahunbay \\
  Department of Computer Science\\
  University of Oxford\\
  Oxford, UK \\
  \And
  Bary Pradelski \\
  Department of Computer Science \\
  CNRS \\
  Grenoble, France \\
  \And
  Marco Scarsini \\
  Dpt. of Econ. and Financial Markets  \\
  Luiss University\\
  Rome, Italy \\
}

\begin{document}

\maketitle

\begin{abstract}
We study equilibrium learning in discretized Bayesian games, focusing on first-price auctions as a testbed for gradient-based learning dynamics under private information. While the symmetric independent private values model admits a unique Bayesian correlated equilibrium (BCE) in the continuous limit, standard convergence guarantees for gradient-based learners do not apply: classical no-regret theory predicts convergence only to Bayesian coarse correlated equilibria (BCCE), which can remain large and far from the competitive equilibrium even under fine discretizations.
To bridge this gap, we introduce the notion of Bayesian semicoarse correlated equilibrium (BSCCE), an equilibrium concept that refines BCCE by restricting deviations to probability-preserving transformations consistent with projected gradient dynamics. We provide a linear programming characterization of BSCCE with polynomially many variables and constraints and establish the inclusion hierarchy $BCE \subseteq BSCCE \subseteq BCCE$. 
Our main finding is that the BSCCE set contracts in many cases as the discretization of the Bayesian game is refined. Even with priors for which BCCE fails to concentrate, the diameter of the BSCCE polytope shrinks toward zero, indicating convergence to the unique continuous-game Bayes-Nash equilibrium. 
These results provide a principled explanation for the observed robustness of gradient-based learning in discretized Bayesian games, showing that these dynamics select a smaller equilibrium set than standard no-regret theory predicts. 
\end{abstract}

\section{Introduction}

Learning in multi-agent systems fundamentally differs from single-agent learning: strategy updates are shaped by the adaptation of other agents. As a result, learning dynamics in games need not converge, may cycle, exhibit chaotic behavior, or converge to different outcomes depending on initial conditions.
Understanding the subset of games where commonly used learning algorithms converge to meaningful equilibrium predictions remains a central challenge at the intersection of machine learning, optimization, and game theory.

Somewhat surprisingly, recent papers observe convergence close to the Bayes-Nash equilibria (BNE) in economically relevant Bayesian games \citep{bichler2023computing, bichler_learning_2021}.
These methods build on gradient learning, and we will focus on its simplest variant, projected gradient ascent, in this paper. 

We will deliberately restrict our attention to first-price single-item auctions. This game is a canonical example used in large-scale applications such as display advertisements. However, despite the uniqueness of the BNE, the empirical convergence of gradient-based algorithms has not yet been fully theoretically justified.
Here, \textit{we aim to provide an explanation for the convergence of projected gradient ascent in first-price auctions under increasingly finer discretizations.} 
The BSCCE, however, is a broader concept and can be applied to all games with private information.

\paragraph{Gradient Dynamics in Games.}

Motivated by their algorithmic simplicity, recent work studies gradient-based learning dynamics in games 
\citep{
mertikopoulos_riemannian_2018, mertikopoulos_unified_2024, cai_fast_2024}.
While it is known that Nash equilibria are equilibria of the resulting dynamical system,
strong assumptions, including strict monotonicity, are usually needed to guarantee uniqueness and the convergence of the dynamics to such an equilibrium in last-iterate, independent of initial conditions \citep{mertikopoulos_learning_2018, zhang_stability_1995}. Similarly, Lyapunov-function-based arguments can be used to show stability \citep{mertikopoulos_learning_2016}. 
In general, the dynamics of a game may also cycle, prohibiting last-iterate convergence statements in general.

Dynamical systems offer a powerful toolset, but the analysis is game-specific and often requires significant effort to use (see, e.g., \cite{bichler_convergence_2024}). 
In Bayesian games with private information, the picture is even less well understood: equilibrium strategies are functions that map types to actions, which further complicates the analysis.
In particular, \cite{bichler_convergence_2024, bichler2025beyond} have shown that the first-price auction satisfies neither the monotonicity nor the (weaker) Minty condition, which rules out many of the related convergence statements. 

\paragraph{No-regret Learning and Bayesian (Coarse) Correlated Equilibria.}

As an alternative to dynamical systems, we draw on no-regret learning theory, shifting the focus from last-iterate to time-average convergence. Indeed, it is well known that no-regret algorithms converge to (Bayesian) coarse correlated equilibria (BCCE) \citep{hartline2015no}. Similarly, no-internal-/no-swap-regret algorithms converge to (Bayesian) correlated equilibria (BCE) \citep{gordon_no-regret_2008}. 
However, the BCCE can be a large superset of the BNEs of a game, and this result says little about the convergence of no-regret algorithms to a BNE. Meanwhile, no-internal regret algorithms are much more complex, and BCEs may be too strict to capture the behavior of gradient learning algorithms.

In the context of auctions, \citet{ahunbay2024uniquenessbayesiancoarsecorrelated} recently showed that the set of BCCEs collapses to the unique BNE in all-pay auctions with continuous types, which proves convergence of no-regret algorithms to an approximate BNE in discretized versions of the auction. 
However, the authors show that this explanation is insufficient in first-price auctions: the BCCE set remains large for non-concave priors even under arbitrarily fine discretizations and can contain outcomes far from the competitive equilibrium. Thus, standard no-regret guarantees do not explain why gradient-based learners consistently converge close to the BNE in this setting.

Our research relates to the work by \citet{bergemann_first-price_2017}
which analyzes BCEs under varying levels of information provided to the bidders. They find a subset relation for the BCEs as information quality increases (modeled as "information structures"), meaning that as information quality increases, the BCE set shrinks. 
Slightly different from our setup, their model assumes that bidders do not observe their private valuations directly, but instead receive a signal chosen by the auctioneer.

\paragraph{Contributions.}

The large support of BCCE warrants a closer inspection of gradient algorithms in first-price auctions.
The central contribution of this paper is to refine the equilibrium notion used to analyze gradient-based learning in Bayesian games. 
We introduce the \emph{Bayesian Semicoarse Correlated Equilibrium (BSCCE)}, which captures the outcomes of gradient dynamics more precisely than prior concepts.
BSCCE is motivated by the observation that gradient-based algorithms can realize only a restricted class of deviations, namely, tangential, probability-preserving transformations of mixed strategies, which correspond to gradients of potential functions.

BSCCE lies between BCE and BCCE as it restricts (expands) the class of deviations relative to BCE (BCCE). \Cref{fig:Polytope_uniform} illustrates the maximal support across equilibrium concepts in a first-price auction.

\begin{figure}[htb]
    \centering
     \begin{subfigure}[t]{0.24\linewidth}
        \centering
        \includegraphics[width=\linewidth]{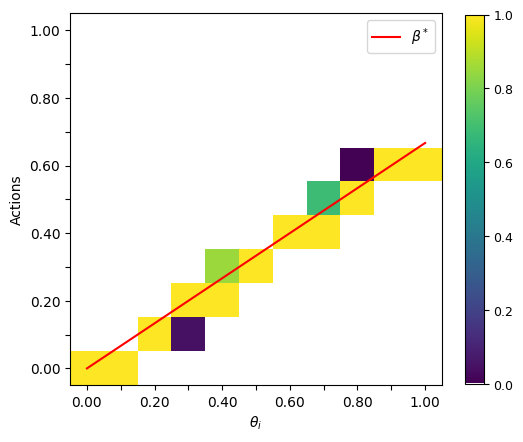}
        \caption{BNE}
        \label{fig:BNE_uniform}
    \end{subfigure}
    \begin{subfigure}[t]{0.24\linewidth}
        \centering
        \includegraphics[width=\linewidth]{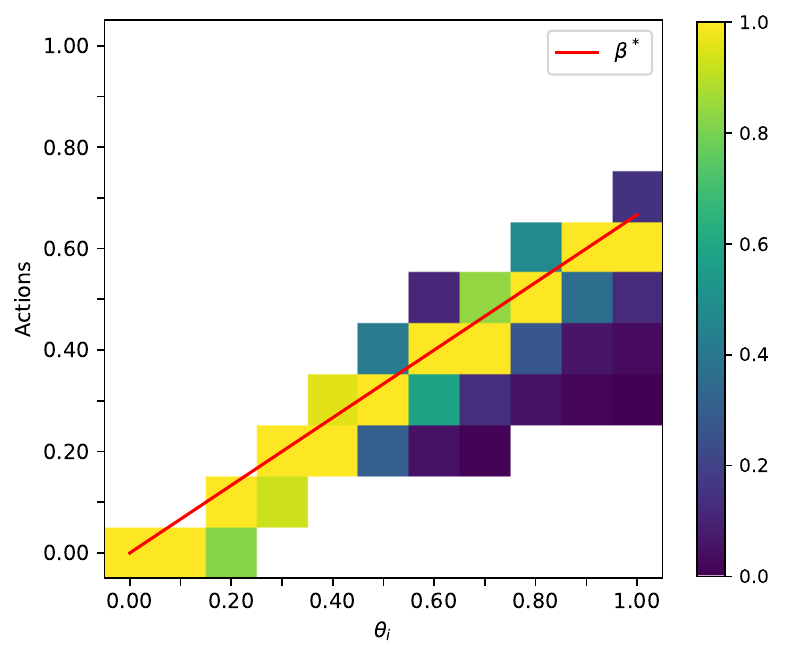}
        \caption{BCE}
        \label{fig:BCE_uniform}
    \end{subfigure}
    \begin{subfigure}[t]{0.24\linewidth}
        \centering
        \includegraphics[width=\linewidth]{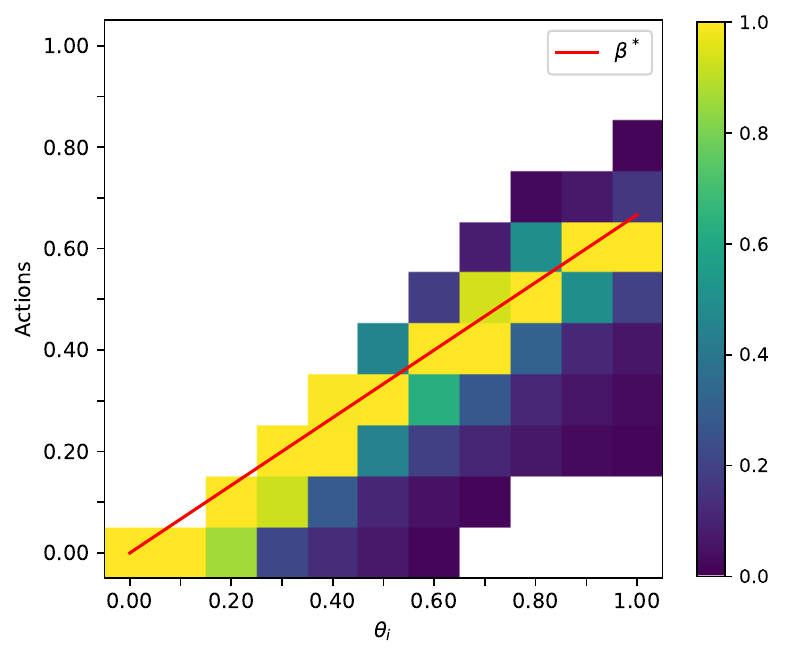}
        \caption{BSCCE}
        \label{fig:BSCCE_uniform}
    \end{subfigure}
    \begin{subfigure}[t]{0.24\linewidth}
        \centering
        \includegraphics[width=\linewidth]{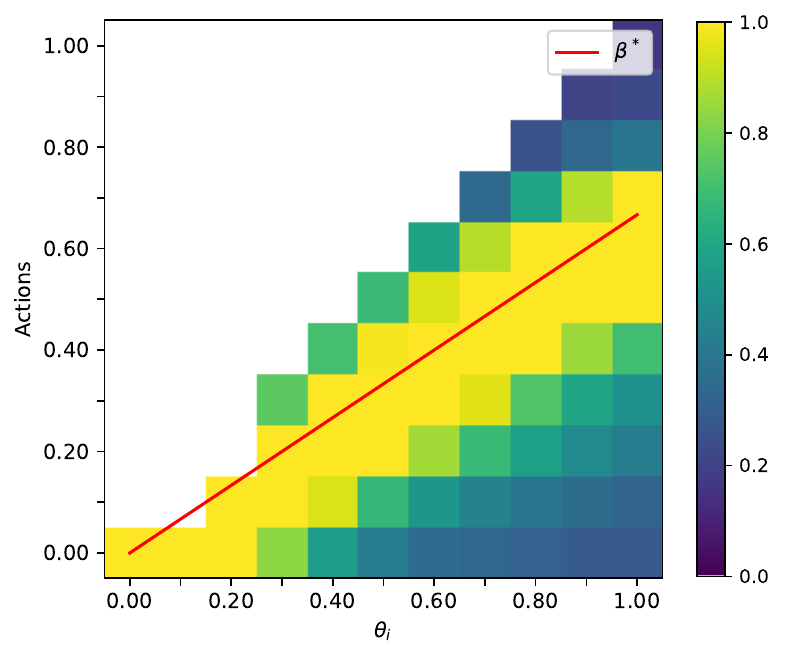}
        \caption{BCCE}
        \label{fig:BCCE_uniform}
    \end{subfigure}

    \caption{Maximal conditional bid probabilities ($\sigma(a_i \mid \theta_i)$) in the equilibrium polytopes for a $\mathrm{Beta}(2,1)$ prior using uniform discretization in a first-price auction.}
    \label{fig:Polytope_uniform}
\end{figure}

Formally, BSCCE extends the notion of Semicoarse Correlated Equilibrium (SCCE), originally developed for complete-information games \citep{ahunbay2025semicoarsecorrelatedequilibrialpbased}, to Bayesian settings. 
However, this extension must be handled carefully, as the strategic normal form of a Bayesian game is exponential in the number of types. Thus, applying the SCCE to the strategic normal form is computationally infeasible.
In contrast, our approach is based on the partial Bayesian approach \cite{forges_five_1993}, which allows a polynomial characterization of BSCCE in terms of the game size and thus makes BSCCE tractable in discretized games.

BSCCE is defined via linear incentive constraints against a class of deviations, restricted to affine, symmetric, probability-preserving deviations that capture the geometry of projected gradient ascent. This approach relates to the concept of $\Phi$-equilibria, as introduced by \citet{GreenwaldJafari2003} and \citet{StoltzLugosi2007}. 
Another related variant of $\Phi$-equilibria are linear correlated equilibria (LCE), for which \cite{Daskalakis-LCE-25} provide a polynomial algorithm that has no-$\Phi$-regret in convex games. The deviations leading to LCEs include those considered in this paper. Therefore, LCEs are a stricter notion of equilibrium ($\mathrm{LCE} \subseteq \mathrm{SCCE}$). We focus on (B)SCCE because of its relation to gradient learners, which has been established by \cite{ahunbay_first-order_2025, ahunbay2025semicoarsecorrelatedequilibrialpbased}.

We apply BSCCEs to discrete first-price auctions with different priors. Our results show convergence of gradient-based algorithms in many first-price auctions as the discretization is refined. This includes auctions with non-concave priors where the BCCE set remains large \citep{ahunbay_first-order_2025}.

Taken together, these results justify the empirical success of gradient-based learning in discretized Bayesian games, selecting a smaller equilibrium set than standard no-regret theory predicts. 
A full theory for the convergence of gradient algorithms remains an open task, but our proposed concept provides important insights on the convergence of gradient-based learning algorithms.

\section{The Model}
We consider strategic interactions under incomplete information. We begin by defining the standard game-theoretic concepts for an incomplete information setting. 
Let $\Delta(\mathcal{H})$ denote the set of probability distributions over a measurable space $(\mathcal{H}, \mathcal{B}(\mathcal{H}))$, where $\mathcal{B}(\mathcal{H})$ is the Borel $\sigma$-algebra associated with the metric space $\mathcal{H}$.

\begin{definition}[Bayesian Game]
\label{def:bayesian_game}
A \emph{Bayesian game} is a tuple $\mathcal{G} = (\Ncal, \Theta, \Acal, u, \pi)$, where 
\begin{itemize}
    \item $\Ncal = \{1, \dots, n\}$ is a finite set of players.
    \item $\Theta = \Theta_1 \times \cdots \times \Theta_n$ is the type space.
    \item $\Acal = \Acal_1 \times \cdots \times \Acal_n$ is the space of action profiles.
    \item $u = (u_1, \dots, u_n)$ are utility functions, where\\
    $u_i: \Theta_i \times \Acal \to \mathbb{R}$ specifies the payoff of player $i$.
    \item $\pi \in \Delta(\Theta)$ is the joint prior distribution of types.
\end{itemize}
\end{definition}
When discussing distributions such as $\sigma(\theta, a)$, we denote the marginal distributions by omitting the marginalized variables. For example, $\sigma(a)$ denotes the marginal distribution over actions. We denote conditional distributions and expectations with vertical bars; for instance, $\sigma(a_{-i} \mid \theta_i, a_i)$ is the conditional probability of the opponents' actions $a_{-i}$ given player $i$'s action $a_i$ and type $\theta_i$. 

A \emph{pure strategy} for player $i$ is a mapping $s_i: \Theta_i \to \Acal_i$. A \emph{mixed strategy} $x_i: \Theta_i \to \Delta(\Acal_i)$ maps types to distributions over actions. 

The expected utility of player $i$ given a strategy profile $x$ and type $\theta_i$ (\emph{interim utility}) is $u_i(x_i, x_{-i} \mid \theta_i)$:
\begin{align*}
    \mathbb{E}_{\theta_{-i}, a \sim x} [u_i(\theta_i, a) \mid \theta_i] =
    \sum_{\theta_{-i} \in \Theta_{-i}} \pi(\theta_{-i} \mid \theta_i) \sum_{a \in \mathcal{A}} \left( \prod_{j \in \mathcal{N}} x_j(a_j \mid \theta_j) \right) u_i(\theta_i, a).
\end{align*}
This definition can be extended to infinite type- and action spaces by replacing the corresponding sums with integrals.
Observe that the interim utility is linear in $x_i$, the player's own action distribution.

\begin{definition}[Bayes-Nash Equilibrium (BNE)]
\label{def:bne}
A strategy profile is a \emph{Bayes-Nash Equilibrium} if it maximizes interim utility for almost all types.
More specifically, a \emph{distributional} profile $x^* = (x^*_1, \dots, x^*_n)$, where $x^*_i: \Theta_i \to \Delta(\mathcal{A}_i)$, is a BNE if for all $i \in \mathcal{N}$ and almost all $\theta_i \in \Theta_i$:
\begin{align*}
    u_i(x^*|\theta_i) \geq u_i(a_i^\prime, x_{-i}^*| \theta_i)
    \quad \forall a'_i \in \mathcal{A}_i.
\end{align*}
\end{definition}

There are multiple approaches to define Bayesian correlated equilibria (BCEs) \citep{forges_five_1993}. We will rely on the "partial Bayesian approach", which assumes that all players observe their types and a recommendation for a corresponding action. In contrast to the "strategic form" and "agentic form" approaches, recommendations are correlated on the level of individual actions and may depend on the realized types of \emph{all} players. For this, we need the following concepts.
Let $\sigma \in \Delta(\Theta \times \mathcal{A})$ denote a distribution over types and actions.
\begin{definition}[Prior Consistency]
    $\sigma$ satisfies \emph{prior consistency} with respect to a prior $\pi$ if the marginal distribution of types induced by $\sigma$ coincides with the exogenous prior $\pi$.
    \[
        \sum_{a \in \mathcal{A}} \sigma(\theta,a) =  \sigma(\theta)= \pi(\theta)
        \quad \forall \theta \in \Theta .
    \]
\end{definition}
\begin{definition}[Conditional Independence]
$\sigma$ satisfies \emph{conditional independence} if
\[
\sigma(\theta_{-i} \mid \theta_i, a_i) = \sigma(\theta_{-i} \mid \theta_i)
\quad \forall i .
\]
\end{definition}
This condition ensures that a player’s recommendation does not reveal information about opponents.

\begin{definition}[BCE and BCCE]\label{def:bce}
A distribution $\sigma \in \Delta(\Theta \times \Acal)$ satisfying prior consistency is

a \emph{Bayesian correlated equilibrium (BCE)} if, for all players $i$, types $\theta_i$, and actions $a_i, a_i'$,
    \begin{equation}
    \label{eq:bce-incentive-compatibility}
    \sum_{\theta_{-i}, a_{-i}}
    \sigma(\theta, a)
    \bigl[
        u_i(a_i, a_{-i}| \theta_i)
        -
        u_i(a_i', a_{-i}| \theta_i)
    \bigr]
    \ge 0 ;
    \end{equation}
    
a \emph{Bayesian coarse correlated equilibrium (BCCE)} if, for all players $i$, types $\theta_i$, and actions $a_i'$,
    \[
    \sum_{\theta_{-i}, a}
    \sigma(\theta, a)
    \bigl[
        u_i(a| \theta_i)
        -
        u_i(a_i', a_{-i}| \theta_i)
    \bigr]
    \ge 0 .
    \]
Optionally, one may impose conditional independence.
\end{definition}

The \emph{Incentive compatibility} constraint (\ref{eq:bce-incentive-compatibility}) for a BCE must hold for each type, and means that no player wants to deviate from their recommended action $a_i$ to another action $a_i'$ after observing their type $\theta_i$.
The hierarchy of the introduced equilibrium concepts is well-known: 
\[
    BNE \subseteq BCE \subseteq BCCE.
\]

Conditional independence is optional across equilibrium concepts. The inclusion hierarchy is preserved whether conditional independence is imposed for all equilibrium concepts or for none. We adopt the former as the default.

\paragraph{Single-Item, First-Price Auctions.}
\label{sec:auctions}
We focus on the standard symmetric Independent Private Values (IPV) model. The set of bidders is $\Ncal = \{1, \dots, n\}$. Each bidder $i$ draws a private valuation $\theta_i$ independently from a probability mass function (or density) $\pi$ with corresponding cumulative distribution function $F$ with support on $[0, \bar{\theta}]$. W.l.o.g. we can set $\bar \theta = 1$.
Bidders submit bids $a_i \in \Acal_i = [0,1]$. The restriction $\Acal_i = [0,1]$ is w.l.o.g. as rational bidders never overbid.

The auction is defined by an allocation rule $q_i(a)$. The utility of bidder $i$ with valuation $\theta_i$ is:
\begin{equation}
    u_i(a| \theta_i) = q_i(a) \cdot (\theta_i - a_i).
\end{equation}
We assume a standard highest-wins allocation rule with uniform tie-breaking:
\begin{equation}
    q_i(a) = \begin{cases} 
        1 & \text{if } a_i > \max_{j \neq i} a_j \\
        \frac{1}{|\{j : a_j = \max a\}|} & \text{if } a_i = \max_{j} a_j \\
        0 & \text{otherwise.}
    \end{cases}
\end{equation}

In the \emph{first-price auction}, the winner pays their bid.

Under regularity assumptions \citep{UniqueEquilibria18}, the first-price auction has a unique BNE $\beta^{\text{FP}}$ in the symmetric IPV model, which is given by the expectation of the second-highest valuation conditioned on the bidder having the highest valuation (i.e., the highest of the others):
\begin{equation}
    \beta_i^{\text{FP}}(\theta_i) = \mathbb{E}[Y_1 \mid Y_1 < \theta_i],
\end{equation}
where $Y_1$ is the highest order statistic of the $n-1$ other bidders' valuations. For example, if valuations are uniformly distributed ($\theta_i \sim U[0,1]$), this simplifies $\beta_i^{\text{FP}}(\theta_i) = \frac{n-1}{n}\theta_i$.

\paragraph{Equilibria in the Discretized First-Price Auction.}
To find equilibria in more general settings, where theoretical predictions are unavailable, learning algorithms discretize the game. However, this discretization alters the game already in the symmetric IPV model and breaks the existence of the unique pure symmetric BNE \citep{rasooly2022importancediscreteinaccuracycontinuous}. Consequently, learning algorithms are no longer expected to converge to a singleton but rather to a set of mixed strategies. 
\Cref{fig:Polytope_uniform} illustrates the set of possible BNEs and BCEs for an example.
We are interested in how the limit set of the learning dynamics, in particular the projected gradient dynamics, changes as the discretization becomes finer.

\section{Bayesian Semicoarse Correlated Equilibrium}

Most learning dynamics are not guaranteed to converge to BNEs but only to relaxed notions thereof. 
On the one hand, BCEs are a relatively tight superset of BNEs, but they lack convergence guarantees for gradient-type algorithms.
On the other hand, the set of BCCEs is too large in the case of First Price auctions to derive meaningful convergence statements \citep{ahunbay2024uniquenessbayesiancoarsecorrelated}.
Therefore, we need to move on to another, relaxed notion of equilibrium.
We propose the Bayesian Semicoarse Correlated Equilibrium (BSCCE) for this purpose, which lies between BCE and BCCE.

\subsection{Semicoarse Correlated Equilibrium (SCCE)}
SCCEs were introduced by \cite{ahunbay2025semicoarsecorrelatedequilibrialpbased} to capture the convergence set of gradient-based learners in complete-information games while maintaining computational feasibility. It can be understood as a specific instance of $\Phi$-equilibria \citep{GreenwaldJafari2003, StoltzLugosi2007} that includes CEs but is (sometimes) strictly smaller than the CCE polytope.
\citet{ahunbay2025semicoarsecorrelatedequilibrialpbased} already applied this solution concept to complete-information versions of Bertrand competition and first-price auctions.

Let us recall the definition of SCCEs and provide an intuitive explanation.
Gradient-based dynamics satisfy no-regret guarantees with respect to a restricted class of deviations that correspond to gradient fields of potential functions \citep{ahunbay_first-order_2025}.
By restricting to gradients of quadratic functions, we obtain a solution concept that is representable as an LP and thus computationally tractable, and it is a superset of the outcomes of gradient dynamics.

A probability distribution $\sigma \in \Delta(\mathcal{A})$ is an SCCE if it is stable against infinitesimal deviations in the direction $Q x_i + q_i$
defined by a pair $(Q_i, q_i)$ for each player $i$, where
\begin{itemize}
    \item $Q_i \in \mathbb{R}^{|\Acal_i| \times |\Acal_i|}$ is a symmetric matrix defining probability flow between actions, and 
    \item $q_i \in \mathbb{R}^{|\Acal_i|}$ is a vector providing an offset.
\end{itemize}
Since the utility function $u_i$ is linear in player~$i$'s action, restricting attention to infinitesimally small deviations entails no loss of generality, as any local deviation can be extended globally by linearity.

These deviations capture the form of update directions that arise under gradient steps which are projected back into the simplex.
Furthermore, to avoid projections in the analysis and to tailor the solution concept to mixed extensions of normal-form games, \cite{ahunbay2025semicoarsecorrelatedequilibrialpbased} restrict the $(Q_i, q_i)$ according to three conditions:

1. \textit{Tangency:} $\sum_{a'_i \in \Acal_i} (Q_i(a'_i, a_i) + q_i(a'_i)) = 0$ for all $a_i$. This ensures probability is merely reallocated between actions, not created or destroyed.
    
2. \textit{Non-negativity:}
$Q_i(a'_i, a_i) + q_i(a'_i) \ge 0$ for all $a'_i \neq a_i$.
This condition ensures that, at pure strategies, admissible deviation directions lie within the tangent cone of the probability simplex, so that infinitesimal deviations preserve feasibility.
    
3. \textit{Symmetry:} 
$Q_i = Q_i^\top$.
This ensures that deviations are gradient fields of quadratic functions. Projected gradient ascent can only move along gradients of scalar functions. Without symmetry, the deviation direction of $Q_i x_i + q_i$ can exhibit rotational flow, which is impossible for gradient fields.

The incentive constraints are infinitesimal deviations in the direction $Q_i x_i + q_i$ evaluated for pure strategies (i.e., $x_i$ is a canonical basis vector).
Thus, the pairs $(Q_i, q_i)$ define deviations $d_i(a', a) = Q_i(a', a) + q_i(a')$ on actions $a', a \in \mathcal{A}_i$. 
We denote the \textit{set of admissible deviations}, i.e., where $d_i$ satisfies the conditions above, by $\mathcal{D}$.

In summary, we get the following definition of SCCEs.
\begin{definition}[Semicoarse Correlated Equilibrium (SCCE), \cite{ahunbay2025semicoarsecorrelatedequilibrialpbased}]
    A distribution $\sigma$ on the joint action profiles is a \emph{semicoarse correlated equilibrium} if, for every agent $i$ and every $d \in \mathcal{D}$, the following incentive compatibility constraint holds:
    \begin{equation*}
        \sum_{a \in A} \sigma(a) \left[ \sum_{a'_i \in \Acal_i} d(a_i', a_i) u_i(a'_i, a_{-i}) \right] \le 0.
    \end{equation*}
\end{definition}
Although there are uncountably many incentive compatibility constraints as the set of pairs $(Q_i, q_i)$ forms a polyhedral cone, this cone admits a finite representation via its extreme rays.
This yields a finite set of constraints, as detailed in \cref{app:LPs}. 
Moreover, the resulting linear program has size polynomial in the cardinalities of the type and action spaces and thus is computationally tractable.

\subsection{Extension to Bayesian Games}
For the extension of SCCEs to Bayesian games, we adopt the \emph{partial Bayesian} framework in line with the above definition of BCEs. This choice is motivated by the structure of learning dynamics in auction games: gradient-based algorithms operate on \emph{type-contingent mixed actions} $x_{i,t}(\cdot \mid \theta_i)$ and update them using local, interim utility gradients. They do not optimize over or maintain distributions over complete strategy functions $s_i : \Theta_i \to \mathcal{A}_i$.
The partial Bayesian formulation mirrors this behavior by defining equilibrium objects directly as joint distributions over types and actions. 

A Bayesian Semicoarse Correlated Equilibrium (BSCCE) is therefore defined as a joint distribution $\sigma \in \Delta(\Theta \times \mathcal{A})$. The definition inherits the structural conditions of Bayesian correlated equilibrium (Definition~\ref{def:bce}), such as prior consistency and conditional independence, but replaces the standard incentive constraints with semicoarse incentive constraints tailored to gradient-based deviations.

We build on the SCCE framework and allow for admissible deviations that are \emph{type-specific}: for each player $i$ and type $\theta_i \in \Theta_i$, deviations are generated by a pair $(Q_i^{\theta_i}, q_i^{\theta_i})$ defining a tangential, probability-preserving transformation of the action simplex. For notational simplicity, we denote such deviations by $d^{\theta_i}(a_i', a_i) = Q_i^{\theta_i}(a_i', a_i) + q_i^{\theta_i}(a_i')$.

\begin{definition}[Bayesian Semicoarse Correlated Equilibrium (BSCCE)]
\label{def:bscce}
A \emph{Bayesian Semicoarse Correlated Equilibrium (BSCCE)} is a joint probability
distribution $\sigma \in \Delta(\Theta \times \mathcal{A})$ that satisfies:

1. \emph{Prior consistency.}

2. \emph{Conditional independence.}

3. \emph{Semicoarse incentive compatibility.}
No player can gain by admissible deviations. Formally, we require that, for every player $i$, every type $\theta_i \in \Theta_i$, and every deviation $d^{\theta_i} \in \mathcal{D}$
\begin{equation*}
\tiny
    \expec_{(\theta_{-i}, a) \sim \sigma}
    \Biggl[
        \sum_{a'_i \in \Acal_i}
        d^{\theta_i}(a'_i,a_i)\,
        u_i(a'_i,a_{-i} \mid \theta_i) \mid
        \theta_i
    \Biggr]
    \le 0.
\end{equation*}
\end{definition}
The conditioning of expectations affects the player's beliefs about other players and their actions.

The BSCCE constraints are all linear. We will add a linear objective that measures the distance to a given distribution (the BNE strategy) and obtain a linear program, which we provide in \Cref{app:LPs}.

\paragraph{Relation to other equilibrium concepts.}
Equilibrium notions can be ordered by the class of admissible deviations they allow. Allowing a richer class of deviations induces more incentive constraints and therefore yields a smaller equilibrium set.
The Bayes-Nash equilibrium (BNE) is the most restrictive notion, as it allows all unilateral deviations at the level of complete strategies and requires that the actions be sampled independently for each player. Bayesian correlated equilibrium (BCE) relaxes this requirement by allowing players to condition deviations on correlated action recommendations, while still enforcing obedience to all such deviations.
BSCCE further relaxes the equilibrium notion by restricting attention to a \emph{proper subset} of deviations. 
Bayesian coarse correlated equilibrium (BCCE) imposes the weakest incentive constraints by allowing only unconditional deviations, and therefore yields the largest equilibrium set. Accordingly, the equilibrium concepts satisfy the inclusion chain
\begin{equation}
\label{eq:subset_relations}
    \text{BNE} \subseteq \text{BCE} \subseteq \text{BSCCE} \subseteq \text{BCCE}.
\end{equation}

\section{Bayesian Gradient Ascent Learners}

To relate BSCCE to learning dynamics, we study gradient-based algorithms in discretized Bayesian games. Our goal is not to model all possible feedback mechanisms, but to identify a class of learning dynamics whose limiting behavior can be characterized by semicoarse incentive constraints.

\paragraph{Expected-Utility Formulation.}
\label{sec:expected_utility_formulation}

A key modeling choice is how learning agents process private information. Rather than modeling stochastic type realizations and noisy gradient feedback, we adopt an \emph{expected-utility formulation} that treats the Bayesian game as a deterministic, high-dimensional game over type-contingent mixed actions.
Specifically, each player $i$ maintains a mixed action $x_i(\cdot \mid \theta_i) \in \Delta(\Acal_i)$ for every type $\theta_i \in \Theta_i$. The player’s decision variable is therefore the vector
\[
x_i \;\in\; X_i := \prod_{\theta_i \in \Theta_i} \Delta(\Acal_i),
\]
and payoffs are given by the ex-ante expected utility
\[
u_i(x) \;=\; \mathbb E_{\theta \sim \pi} \mathbb E_{a \sim x(\cdot \mid \theta)}[u_i(a| \theta_i)].
\]
Under this formulation, the Bayesian game can be viewed as a complete-information normal-form game with a structured action space $X_i$. The utility functions are smooth and multilinear in the mixed actions, and gradients with respect to $x_i(\cdot\mid\theta_i)$ are well defined.

\paragraph{Bayesian Projected Gradient Ascent.}

We consider players who apply projected gradient ascent to maximize their ex-ante expected utility. Given a step size sequence $(\eta_t)_{t\ge 1}$, the update rule for player $i$ is
\begin{equation}
x_{i,t+1}
=
\Pi_{X_i}
\bigl[
x_{i,t} + \eta_t \nabla_{x_i} u_i(x_{1,t}, \dots, x_{n,t})
\bigr],
\end{equation}
where $\Pi_{X_i}$ denotes Euclidean projection.
Because the feasible set $X_i$ is a Cartesian product of simplices, the projection decomposes across types.
We assume that step-sizes satisfy $\eta_t = \omega(1/t)$ and $\sum_{t=1}^T \eta_t = o(T)$.\footnote{Both requirements are in Bachmann-Landau notation.}

This deterministic gradient dynamic avoids stochasticity from sampled type realizations and allows us to directly apply no-regret and stability arguments developed for gradient-based learning in complete-information games.

\paragraph{Connection to Semicoarse Equilibria.}
The expected-utility view of our Bayesian game allows us to extend the convergence result of \citet[Proposition~2.1]{ahunbay2025semicoarsecorrelatedequilibrialpbased}, which shows that projected gradient ascent converges to the set of SCCEs in time-average under complete information, to our incomplete-information setting. 



\begin{corollary}[BPGA implies BSCCE for time-average limit points]
\label{prop:bpga_implies_bscce}
Consider a finite Bayesian game with finite type sets $\Theta_i$ and finite action sets $\Acal_i$.
Let each player $i$ control a type-contingent mixed action vector $x_i \in X_i$ and let the (ex-ante) expected utility be $u_i(x)$. 
Let $(x_t)_{t\ge 1}$ be generated by Bayesian projected gradient ascent (BPGA). 
Define the time-average distribution on $\Theta \times A$ by
\begin{equation}
\label{eq:sigma_T_def}
\sigma_T(\theta,a)
:= \frac{1}{T}\sum_{t=1}^T \pi(\theta)\prod_{j\in\Ncal} x_{j,t}(a_j \mid \theta_j).
\end{equation}
Then every limit point $\sigma$ of $(\sigma_T)_{T\ge 1}$ is a BSCCE.
\end{corollary}

We provide a proof in \cref{sec:proof_bgpa_implies_bscce3}. This establishes BSCCE as a natural equilibrium concept for characterizing the long-run behavior of gradient-based learners in discretized Bayesian games.

\paragraph{Stochastic learning and sampled feedback.}
The BPGA dynamics studied above should be understood as an analytical proxy for the stochastic learning algorithms commonly used in practice. 
In typical auction-learning implementations, agents observe batched realized types and payoffs and update their mixed strategies using stochastic gradient estimates rather than exact interim gradients. 
Under standard assumptions, such stochastic updates track the corresponding expected-utility gradient dynamics in the limit
\citep{Benveniste-Stoch-Approx-12}.
This suggests that BSCCE remains the correct limiting object for stochastic gradient dynamics. However, we do not claim that the stochastic case is fully resolved here. A full stochastic-gradient analysis of auction games is an important next step, but it is beyond the scope of this paper.

\section{Convergence Analysis}
\label{sec:convergence_analysis}

In this section, we characterize the size of the Equilibrium sets for first-price auctions with two players in the IPV model. 
Since we study discretized auctions which have an ordered action space, we are not only interested in how much probability can be shifted between two feasible equilibria, but also how far this probability can be shifted.


\paragraph{Metrics.}
Given a possibly mixed strategy $\sigma \in \Delta(\Theta \times \mathcal A)$ and a pure strategy $\beta^*: \Theta \to \mathcal{A}$,
their squared Wasserstein-2 distance is
\[
W_2^2(\sigma,\beta^*)
= \frac{1}{|\mathcal{N}|}\sum_{\theta, a} \sigma(\theta, a) ||a - \beta^*(\theta)||_2^2 .
\]
In the following, we will simply refer to this metric as the Wasserstein distance.
We consider the squared Wasserstein-2 distance between the buyers’ marginal bid distributions and the equilibrium as a suitable objective for our analysis:
This metric captures deviations in the probability mass assigned to different type--action pairs and weighs these based on the distance to the canonical equilibrium $\beta^*$. 
Furthermore, it provides a natural upper bound on the size of equilibrium polytopes $(\mathcal{P})$:
\[
\diameter_\mathcal{E}(\mathcal{P}) := \sup_{\sigma_1, \sigma_2 \in \mathcal{P}} \sqrt{W_2^2(\beta_1,\beta_2)} \leq 2 \sup_{\sigma \in \mathcal{P}} \sqrt{W_2^2(\sigma, \beta^*)}.
\]

As shown by \citet{ahunbay2024uniquenessbayesiancoarsecorrelated}, the BCCE in first-price auctions is unique whenever the prior cumulative distribution function (CDF) is strictly concave. As $\mathrm{BSCCE} \subseteq \mathrm{BCCE}$, their contraction result immediately carries over to BSCCE. Thus, the theoretically interesting case for our analysis is when the conditions guaranteeing BCCE contraction are violated.

Their numerical experiments indicate that even under a uniform prior, which does not satisfy strict concavity, the BCCE that maximizes Wasserstein distance to the canonical equilibrium approaches that equilibrium as the discretization is refined. At the same time, they provide counterexamples in which this distance increases when the prior is no longer concave. In particular, such behavior arises for $\mathrm{Beta}(\alpha,1)$ priors, whose CDFs are strictly convex for $\alpha > 1$.

These findings motivate a closer investigation of whether BSCCE provides a sharper prediction in precisely these regimes. Specifically, we examine whether the Wasserstein distance between BSCCEs and the canonical equilibrium decreases as the discretization becomes finer even when the prior CDF is convex and BCCE need not concentrate.



\paragraph{Discretization.}
We discretize the type and action space uniformly as $\left\{0, 1/\xi, \dots, 1\right\}$, with a size of $\xi+1$.
The probability of a discrete type $\theta_i =\frac{k}{\xi}$ is $F(\frac{k+1}{\xi+1}) - F(\frac{k}{\xi+1})$.
For a $\mathrm{Beta}(\alpha,1)$ prior, the resulting type probabilities become increasingly imbalanced as $\alpha$ grows. In particular, the prior mass concentrates near the highest valuation, leading to $\pi_i(1) \gg \pi_i(0)$ for large $\alpha$. 
Similarly, the probability mass concentrates around $\frac{1}{2}$ for a $\mathrm{Beta}(\alpha,\alpha)$ prior.
This imbalance induces numerical instabilities in equilibrium computation and makes experiments for larger discretizations challenging.


\noindent
\begin{minipage}[htb]{0.55\textwidth}
\paragraph{Results. }
As the discretization refines, the BSCCE polytope contracts. In particular, its diameter, bounded by the maximal Wasserstein distance to the canonical equilibrium, converges to zero, indicating that all BSCCEs approach the unique continuous BNE. This contraction implies that any equilibrium selected by gradient-based learning dynamics in discretized auctions must lie close to the canonical equilibrium, providing strong numerical evidence for convergence in the continuous limit.
The decrease in squared Wasserstein distance is illustrated in \cref{fig:equilibrium-discretization-beta21} for a $\mathrm{Beta}(2,1)$ prior. 
In this case, the Wasserstein distance decreases roughly as $\frac{1}{10 \xi}$.
\end{minipage}
\hfill
\begin{minipage}[htb]{0.42\textwidth}
    \centering
    \includegraphics[width=\linewidth]{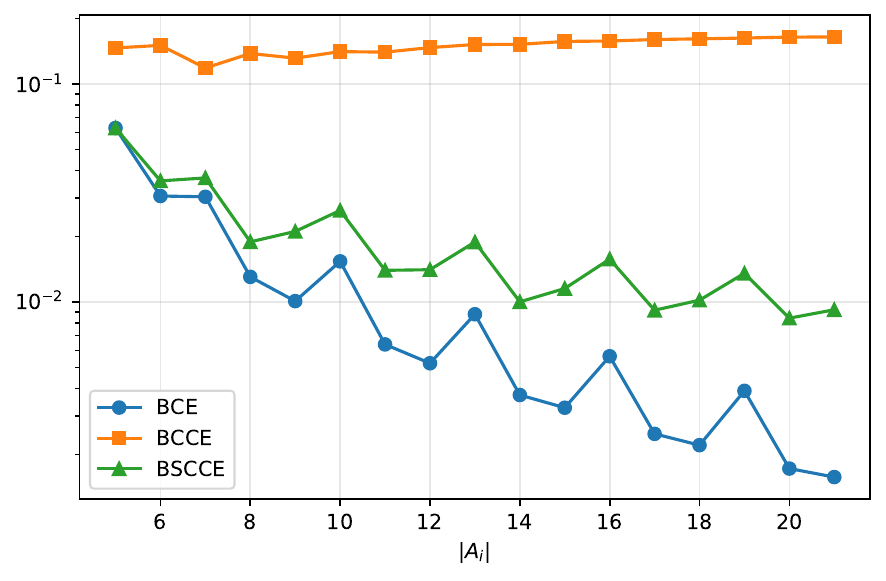}
    \captionof{figure}{Distance to the BNE depending on $\xi$ for a $\mathrm{Beta}(2,1)$ prior.}
    \label{fig:equilibrium-discretization-beta21}
\end{minipage}

\Cref{fig:Polytope_uniform} from our introduction visualizes the maximal conditional probability assigned to each type--action pair under different solution concepts for this prior. Specifically, each tile $(\theta_i, a_i)$ displays $\max_{\sigma \in \mathcal{P}} \sigma(a_i \mid \theta_i)$.
The figures illustrate the inclusion relationship $BNE \subseteq BCE \subseteq BSCCE \subseteq BCCE$. While all solution concepts are centered around the canonical equilibrium $\beta_i^*(\theta_i) = \tfrac{2}{3}\theta_i$, both BSCCE and BCE are substantially more concentrated around $\beta_i^*$ than BCCE. 

The speed of the concentration depends on the priors (see \cref{app:numerics}). When the game is close to a complete-information game, or the CDF is too convex, the size of the BSCCE polytope shrinks slowly. 
We observe that the BSCCE is more robust to the non-concavity of the CDF than the BCCE, but for large values of $\alpha$ (where the convexity strengthens), it is not able to track the BCE convergence.



\paragraph{Complexity.} An implementation of the general BSCCE linear program for a discretization of $1/\xi$ involves $O(\xi^4)$ variables and $O(\xi^3)$ constraints, with a relatively dense constraint matrix. Several model reductions are available when players are symmetric, and additional reductions apply to first-price auctions, where rational players never bid above their type. This allows us to reduce the number of non-zeros in the LP by approximately an order of magnitude. Even after these optimizations and despite the polynomial scaling, the LPs become challenging to solve for very fine discretizations.

For a discretization of $\xi=20$, the BSCCE LP contains approximately 45k variables, 8k constraints, one million non-zeros, and requires about 20 seconds to compute the Wasserstein-distance-maximizing BSCCE using the Gurobi Optimizer version 12.1. By comparison, without model reductions, the same instance requires 30 minutes. For $\xi = 30$, the BSCCE LP grows to roughly 180k variables, 25k constraints, and 7 million non-zeros, with a solution time of about 11 minutes.

While we can ensure a theoretically well-behaved complexity bound, this does not always imply practical tractability.
Due to the increasingly numerical instability and increasing computational burden, we did not pursue even finer discretizations.
We still believe that our study here is informative, as it points towards a theoretically justified convergence of gradient-based algorithms.

\section{Limitations}
Our results have some limitations. First, while BSCCE is defined for general finite Bayesian games, our numerical analysis focuses on discretized symmetric first-price auctions. In particular, the contraction of BSCCE toward the continuous-game BNE outside regimes already covered by existing BCCE results is supported numerically rather than established theoretically. When existing results imply that BCCE contracts to the BNE, the inclusion $\mathrm{BCE}\subseteq\mathrm{BSCCE}\subseteq\mathrm{BCCE}$ immediately yields the same conclusion for BSCCE. Beyond these regimes, however, deriving conditions under which BSCCE contracts to the BNE remains an open theoretical question.

Second, our formal learning result concerns time averages of deterministic Bayesian projected gradient ascent using expected-utility gradients. It therefore does not establish convergence guarantees for stochastic gradient updates based on sampled auction outcomes. While stochastic-approximation arguments motivate a connection between these settings, a formal treatment of stochastic BPGA requires additional assumptions and is beyond the scope of this paper.

Third, our numerical evidence is necessarily limited by the computational cost of optimizing over the BSCCE polytope. This constrains the discretization sizes that can be studied. Moreover, uniform discretizations can generate grid-alignment effects because the continuous BNE may lie closer to the action grid for some grid sizes than for others. We therefore interpret individual finite-grid results with caution and focus on the behavior across increasingly fine discretizations.

Finally, the auction experiments considered here use symmetric priors. Extending the analysis to asymmetric priors is nontrivial because the corresponding equilibrium is asymmetric and generally does not admit the same closed-form benchmark. Similarly, although the definition of BSCCE and its connection to BPGA apply to finite Bayesian games more generally, whether BSCCE provides a comparably sharp equilibrium prediction in other Bayesian games remains an open empirical and theoretical question.

\section{Conclusion}

This paper revisits a recurring empirical observation in auction learning: gradient-based algorithms applied to discretized Bayesian auctions reliably converge to outcomes that are close to the unique Bayes–Nash equilibrium (BNE) of the underlying continuous game, even in settings where standard theoretical guarantees fail. While existing no-regret arguments based on Bayesian coarse correlated equilibrium (BCCE) explain convergence in all-pay auctions, they are insufficient for first-price auctions, where the BCCE set can remain large and include outcomes far from the BNE.

To address this gap, we introduced the \emph{Bayesian semicoarse correlated equilibrium} (BSCCE), a refinement of BCCE that restricts admissible deviations to tangential, probability-preserving transformations consistent with projected gradient dynamics. This equilibrium concept admits a linear programming characterization that enables systematic numerical analysis.
While the focus of this paper is first-price auctions with independent private values, the solution concept BSCCE is defined for all Bayesian games.

Our main finding is that, in symmetric first-price auctions, the BSCCE set contracts as the discretization is refined. Across a range of priors, including cases where BCCE fails to concentrate, the diameter of the BSCCE polytope shrinks to zero, indicating convergence to the unique continuous BNE. This shrinkage property provides a stability certificate for gradient-based learning dynamics. 
However, shrinkage is limited when the priors are degenerate, i.e., when the game is close to a complete-information game, or when the CDF is highly convex. While the proposed equilibrium concept, BSCCE, sharpens predictions about the outcome of gradient dynamics, it remains open whether there is an even smaller, tractable solution concept that captures that outcome better.

Overall, our results suggest that equilibrium concepts tailored to the structure of learning dynamics can offer sharper and more robust predictions than generic no-regret notions in Bayesian games. 
The BSCCE framework provides a principled bridge between gradient-based optimization and equilibrium analysis, and it is applicable beyond auctions to other multi-agent learning problems with private information and high-dimensional strategy spaces.

\vfill
\pagebreak

\begin{ack}
This project has received funding from the European Research Council (ERC) under the European Union’s Horizon Europe research and innovation programme (grant agreement No 101198689). 

This project was funded by the Deutsche Forschungsgemeinschaft (DFG, German Research Foundation) - GRK 2201/2 - Project Number 277991500.

The work of Mete Şeref Ahunbay was funded by a Walter Benjamin Fellowship from the German Research Foundation (Deutsche Forschungsgemeinschaft), project nr. 558988238.

This project has received funding from the French National Research Agency (ANR) through the PEPR IA FOUNDRY project (ANR-23-PEIA-0003).
\end{ack}


{
\small
\bibliography{bibliography}
\bibliographystyle{icml2026}
}


\newpage
\appendix

\section{LP Formulations}\label{app:LPs}
\subsection{BSCCE LP}
We state BSCCEs as a linear program over the joint distribution $\sigma \in \Delta(\Theta \times \mathcal{A})$ where the constraints enforce the conditions in Definition \ref{def:bscce} (prior consistency, conditional independence, incentive compatibility). Note that $\pi(\theta)$ represents the probability density of the type distribution.

For the conditional independence, we write it explicitly in terms of the joint distribution and already assume prior consistency 
\[
 \sigma(\theta_{-i}|\theta_i, a_i) = \sigma(\theta_{-i}|\theta_i) 
 \Leftrightarrow 
 \frac{\sigma(\theta, a_i)}{\sigma(\theta_i, a_i)} = \frac{\sigma(\theta)}{\sigma(\theta_i)} = \frac{\pi(\theta)}{\pi(\theta_i)} 
 \Leftrightarrow 
 \pi(\theta_i) \sigma(\theta, a_i) = \pi(\theta)\sigma(\theta_i, a_i)
\]

Besides Prior Consistency
\[
    \sum_{a} \sigma(\theta, a) = \pi(\theta)
\]
it must hold for each player $i \in \mathcal{N}$ and all corresponding types $\theta_i\in \Theta_i$ that
\begin{align*}
    \pi(\theta_i) \sum_{a_{-i}}\sigma(\theta, a_i, a_{-i}) &= \pi(\theta) \sum_{\theta_{-i}', a_{-i}}\sigma(\theta_i, \theta_{-i}', a_i, a_{-i}) \quad \forall \; a_i, \theta_{-i}, \tag{Conditional Independence}\\
    \sum_{\theta_{-i}, a } \sigma(\theta, a)
  \left[
     \sum_{a_i'\in \Acal_i} d_i(\theta_i, a_i',a_i) u_i(a_i',a_{-i}| \theta_i)
  \right] &
  \ \le\ 0
  \tag{Semicoarse incentive compatibility}
\\ \forall\, d_i(\theta_i, a_i',a_i) \in \mathcal D_i. 
\end{align*}

The suitable BSCCE deviations $d_i(\theta_i, a_i', a_i) = Q_i(\theta_i, a_i', a_i) + q_i(\theta_i, a_i')\in \mathcal D_i$ are defined as follows:\\
for all players $i\in \mathcal{N}$ and for all corresponding types $\theta_i\in \Theta_i$ it must hold
\begin{align*}
    \sum_{a_i' \in \Acal_i} d_i(\theta_i, a_i', a_i) = 0                    &&\forall a_i ,
    \tag{probability preservation, $\mu_i(\theta_i, a_i)$}\\
     d_i(\theta_i, a_i', a_i) \geq 0                                    &&\forall a_i \neq a_i',
     \tag{non-negativity, $\gamma_i(a_i, a_i^\prime)$}\\
    Q_i(\theta_i, a_i', a_i) - Q_i(\theta_i, a_i, a_i') =0              &&\forall a_i, a_i'.  
    \tag{symmetry, $\rho_i(a_i, a_i^\prime)$}
\end{align*}

\subsection{Tractable BSCCE LP}
\label{app:tractable_BSCCE_LP}

As there are infinitely many 'suitable' deviations, the primal before is numerically unsolvable. We use the approach from Theorem 3.4 in \cite{ahunbay2025semicoarsecorrelatedequilibrialpbased} to get a tractable version.
The idea is to consider the 'worst-case' deviations, which correspond to the extreme rays of the deviation polyhedral cone,
and incorporate them into the objective. This results in an inner and outer problem. Then we consider the dual of the inner maximization problem, with the dual variables $\mu, \rho$ and $\gamma \geq 0$ as introduced before.
The dual is a feasibility problem as the right-hand sides are $0$, and the constraints can be incorporated into the outer problem.
As $\mu_i(\theta_i, a_i)$ will only appear in one equality constraint, it can be eliminated from the resulting LP by fixing its value suitably.
The Semicoarse incentive compatibility constraint is then replaced by the following two types of tractable constraints (for each player and type)
\begin{align*}
    w_i(\theta_i, a_i, a_i^\prime) - w_i(\theta_i, a_i, a_i) + \gamma_i(\theta_i, a_i, a_i^\prime) + \rho_i(\theta_i, a_i, a_i^\prime) - \rho_i(\theta_i, a_i^\prime, a_i) &= 0
   &&\forall a_i\ne a_i^\prime\\
   \sum_{a_i}[ w_i(\theta_i, a_i, a_i^\prime) - w_i(\theta_i, a_i, a_i)] + \sum_{a_i\ne a_i^\prime}\gamma_i(\theta_i, a_i, a_i^\prime) &= 0
   &&\forall a_i^\prime,
\end{align*}
where
\[
w_i(\theta_i, a_i, a_i^\prime) = \sum_{\theta_{-i}, a_{-i}}\sigma(\theta, a) u_i(a_i',a_{-i}| \theta_i).
\]
Note that if $\rho=0$, then the constraints reduce to the BCE incentive constraints.

If all players have identical type and action spaces, i.e.\ $|\Theta_i| = |\Theta_j|$ and
$|\Acal_i| = |\Acal_j|$ for all $i\neq j \in \Ncal$, the resulting LP has
\(
O\!\left((|\Theta_i|\,|\Acal_i|)^{|\Ncal|}\right)
\)
variables and
\(
O\!\left(|\Ncal|\,|\Theta_i|^{|\Ncal|}\,|\Acal_i| + |\Ncal|\,|\Theta_i|\,|\Acal_i|^2\right)
\)
constraints, where the first term comes from the prior-consistency and conditional-independence
constraints, and the second term from the tractable semicoarse incentive-compatibility
constraints.
\section{Proof of \cref{prop:bpga_implies_bscce}}
\label{sec:proof_bgpa_implies_bscce3}

\begin{proof}
  
    We show that the distribution
    \[
    \sigma_T(\theta,a)
    := \frac{1}{T}\sum_{t=1}^T \pi(\theta)\prod_{j\in\mathcal N} x_{j,t}(a_j \mid \theta_j)
    \]
    satisfies the constraints in Definition~\ref{def:bscce}.
    
    \paragraph{Prior consistency.}
    We must show that the marginal of $\sigma_T$ over actions equals the prior $\pi$:
    \(
    \sigma_T(\theta) = \pi(\theta)
    \).
    It holds that,
    \[
    \begin{aligned}
    \sigma_T(\theta) =\sum_{a\in A} \sigma_T(\theta,a)
    &= \sum_{a\in A} \frac{1}{T}\sum_{t=1}^T \pi(\theta)\prod_{j\in\mathcal N} x_{j,t}(a_j \mid \theta_j) \\
    &= \frac{1}{T}\sum_{t=1}^T \pi(\theta)
        \sum_{a\in A} \prod_{j\in\mathcal N} x_{j,t}(a_j \mid \theta_j) = \frac{1}{T}\sum_{t=1}^T \pi(\theta)
    = \pi(\theta).
    \end{aligned}
    \]
    The second-to last step follows, as $x_{j,t}(\cdot\mid\theta_j)$ are a probability
    distribution over $A$, hence the inner sum equals $1$.
    
    


    \paragraph{Conditional independence.}
    We first show that $\sigma_T(a_i \mid \theta_i) = \frac{1}{T} \sum_{t = 1}^{T} x_{i,t}(a_i \mid \theta_i) \quad (*)$.\\
    As a first step, we use Bayes' rule and the definition of $\sigma_T$. Afterwards, we rearrange terms to obtain
    \begin{equation*}
    \tiny
        \begin{aligned}
            \sigma_T(a_i \mid \theta_i) 
            &=  \sum_{a_{-i}, \theta_{-i}} \frac{\sigma_T(a, \theta)}{\sigma_T(\theta_i)}
            &&= \frac{1}{\sigma_T(\theta_i)} \cdot \sum_{a_{-i}, \theta_{-i}} \frac{1}{T} \sum_{t = 1}^{T} \pi(\theta) \prod_{j \in \Ncal} x_{j,t}(a_j \mid \theta_j) \\
            &&&= \underbrace{\frac{\pi(\theta_i)}{\sigma_T(\theta_i)}}_{=1} \cdot \frac{1}{T} \sum_{t = 1}^{T} x_{i,t}(a_i \mid \theta_i) 
            \sum_{\theta_{-i}} \pi(\theta_{-i} \mid \theta_i) \sum_{a_{-i},}\prod_{j \neq i} x_{j,t}(a_j \mid \theta_j)
            &=  \frac{1}{T} \sum_{t = 1}^{T} x_{i,t}(a_i \mid \theta_i).
        \end{aligned}
    \end{equation*}
    The last step is a result of prior consistency and the insight that $\pi(\theta_{-i} \mid \theta_i)$ is a distribution over $\Theta_{-i}$ and $\prod_{j\neq i} x_j(\cdot \mid \theta_j)$ is a distribution over $\Acal_{-i}$, so they sum up to one.

    We now apply this to show conditional independence.
    \begin{equation*}
        \begin{aligned}
            \sigma_T(a_i, \theta) 
            &= \sum_{a_{-i}} \sigma_T(a, \theta)
            &&= \sum_{a_{-i}} \frac{1}{T}\sum_{t=1}^T \pi(\theta)\prod_{j\in\Ncal} x_{j,t}(a_j \mid \theta_j) \\
            &&&= \pi(\theta) \cdot \frac{1}{T}\sum_{t=1}^T x_i(a_i \mid \theta_i) \sum_{a_{-i}} \prod_{j \neq i} x_{j,t}(a_j \mid \theta_j) \\
            &&&= \pi(\theta) \cdot \frac{1}{T}\sum_{t=1}^T x_i(a_i \mid \theta_i) 
            \overset{(*)}{=} \pi(\theta) \cdot \sigma_T(a_i \mid \theta_i)
        \end{aligned}
    \end{equation*}

    \paragraph{Incentive compatibility.} 
    We need to show that for every player $i$, every type $\theta_i \in \Theta_i$, and every deviation $d^{\theta_i} \in \mathcal{D}$
    \begin{equation*}
    \lim_{T\rightarrow \infty}
        \expec_{(\theta, a) \sim \sigma_T}
        \Biggl[
            \sum_{a'_i \in \Acal_i}
            d^{\theta_i}(a'_i,a_i)\,
            u_i(a'_i,a_{-i} \mid \theta_i) \bigg|
            \theta_i
        \Biggr]
        \le 0.
    \end{equation*}
    
    Remember that the players play the strategy game in expected-utility formulation as described in \cref{sec:expected_utility_formulation}.
    In particular, the ex-ante utility for playing a vector $x_i \in X_i := \prod_{\theta_i \in \Theta_i} \Delta(\Acal_i)$ is
    \begin{equation*}
        u_i(x_i, x_{-i}) = \mathbb{E}_{\theta \sim \pi, a \sim x(\cdot \mid \theta)} [u_i(a_i, a_{-i})] 
        = \sum_{\theta, a} \left( \pi(\theta) \cdot \prod_{j \in \Ncal} x_j(a_j \mid \theta_j) \right) ~ u_i(a_i, a_{-i} \mid \theta_i).
    \end{equation*}
    
    This is essentially a complete-information game, and we can apply Proposition 2.1 from \cite{ahunbay2025semicoarsecorrelatedequilibrialpbased}. The proposition states that, for any tangential $h: X \to \mathbb{R}$ with bounded, Lipschitz continuous gradients, the following inequality holds, assuming all players implement projected gradient ascent with suitably chosen step sizes.
    \begin{equation*}
        \frac{1}{T} \sum_{t = 1}^{T} \sum_{i \in \Ncal} \langle \nabla_i h(x_t), \nabla_i u_i(x_t) \rangle \leq \frac{1}{T} \cdot \left( \frac{1}{\eta_T} + \frac{1}{\eta_1} + \sum_{t = 1}^{T} \eta_t \right) \cdot const \quad  \stackrel{T \to \infty}{\longrightarrow} 0
    \end{equation*}
    The constant depends on the game and the function $h$. We choose 
    \begin{equation*}
        h(x) = \sum_{i} \sum_{\theta_i} \frac{1}{\pi(\theta_i)} \left( \frac{1}{2} x_i[\theta_i]^T Q_i^{\theta_i} x_i[\theta_i] + x_i[\theta_i]^T q_i^{\theta_i} \right)
    \end{equation*}
    with pairs $(Q_i^{\theta_i}, q_i^{\theta_i})$ that satisfy the tangency, non-negativity, and symmetry conditions.
    $x_i[\theta_i] \in \Delta(\mathcal{A}_i)$ is the probability vector that player $i$ assigns upon observing $\theta_i$. With this $h$, we now show that BSCCE incentive compatibility follows. We can assume that $\pi(\theta_i) > 0$ without loss of generality because if this were not the case, we could remove the type from the game.
    
    Let us first describe the gradients $\nabla_i h(x)$ and $\nabla_i u_i(x)$. They both are vectors in $\mathbb{R}^{\abs{\Theta_i} \cdot \abs{\Acal_i}}$, and, like the strategy vectors $x_i$, they have a block-wise structure: There are $\abs{\Theta_i}$ blocks with $\abs{\Acal_i}$ entries, corresponding to the gradients with respect to $x_i(a_i \mid \theta_i)$. In particular, we find
    \begin{align*}
        \nabla_i h(x)[\theta_i, a_i'] &= \frac{1}{\pi(\theta_i)} \left( \sum_{a_i} Q_i^{\theta_i}(a_i', a_i) x_i[\theta_i, a_i] \right) +  \frac{1}{\pi(\theta_i)} q_i^{\theta_i}(a_i') \\
        \nabla_i u_i(x)[\theta_i, a_i'] &= u_i(a_i', x_{-i} \mid \theta_i) \cdot \pi(\theta_i),
    \end{align*}
    where the indices refer to the position in the vector corresponding to the respective type and action.
    The utility $u_i(a_i', x_{-i} \mid \theta_i)$ is the interim utility of a player that observes type $\theta_i$ and plays action $a_i'$ while the opponents commit to playing $x_{-i}$:
    \begin{align*}
        u_i(a_i', x_{-i} \mid \theta_i) &= \mathbb{E}_{\theta_{-i} \sim \pi(\cdot \mid \theta_i)} \left[ \mathbb{E}_{a_{-i} \sim x_{-i}(\cdot \mid \theta_{-i})} [u_i(a_i', a_{-i} \mid \theta_i)] \right]\\
        &= \sum_{a_{-i}, \theta_{-i}} \left( \pi(\theta_{-i} \mid \theta_i) \prod_{j \neq i} x_j(a_j \mid \theta_j) \right) ~ u_i(a_i', a_{-i} \mid \theta_i).
    \end{align*}
    
    In what follows, we derive the equality 
    \begin{equation*}    
    \frac{1}{T} \sum_{t = 1}^T \sum_{i \in \Ncal} \left\langle \nabla_i h(x), \nabla_i u_i(x) \right\rangle = \sum_{i \in \Ncal} \sum_{\theta_i \in \Theta_i} \mathbb{E}_{(\theta, a) \sim \sigma_T} \left[ \sum_{a_i'} d^{\theta_i}(a_i', a_i) u_i(a_i', a_{-i} \mid \theta_i) \mid \theta_i \right], 
    \end{equation*}
    which allows us to apply the bound from the proposition to our distribution $\sigma_T$.

    We start with the scalar product of $\nabla_i u_i$ and $\nabla_i h$:
    \begin{equation*}
    \label{eq:gradient_scalar_product}
    \begin{aligned}
        \left\langle \nabla_i u_i(x_t), \nabla_i h(x_t) \right\rangle 
        &= \sum_{a_i', \theta_i} \nabla_i u_i(x_t)[\theta_i, a_i'] \cdot \nabla_i h(x_t)[\theta_i, a_i']  \\
        &= \sum_{a_i', \theta_i} u_i(a_i', x_{-i,t} \mid \theta_i) \cdot \left( \left( \sum_{a_i}  Q_i^{\theta_i}(a_i', a_i) \cdot x_{i,t}(a_i \mid \theta_i) \right) + q_i^{\theta_i}(a_i') \right) \\
        &= \sum_{a_i', \theta_i} u_i(a_i', x_{-i,t} \mid \theta_i) \sum_{a_i} \left(  Q_i^{\theta_i}(a_i', a_i) + q_i^{\theta_i}(a_i') \right) \cdot x_{i,t}(a_i \mid \theta_i)
    \end{aligned}
    \end{equation*}
    In the second step, we used our gradients and note that the $\pi(\theta_i)$ terms cancel. We also moved from the vector-based notation for the players' strategies to a conditional-probability notation. In the second step, we applied the identity $\sum_{a_i} x_{i,t}[\theta_i, a_i] = 1$.

    We expand the interim utility terms to obtain
    \begin{equation*}
        \small
        \begin{aligned}
            &= \sum_{a_i', \theta_i} \underbrace{\sum_{a_{-i}, \theta_{-i}} \left( \pi(\theta_{-i} \mid \theta_i) \prod_{j \neq i} x_{j,t}(a_j \mid \theta_j) \right) \cdot u_i(a_i', a_{-i} \mid \theta_i)}_{u_i(a_i', x_{-i,t} \mid \theta_i)} \cdot \sum_{a_i} \left(  Q_i^{\theta_i}(a_i', a_i) + q_i^{\theta_i}(a_i') \right) \cdot x_{i,t}(a_i \mid \theta_i) \\ 
            &= \sum_{a_i', \theta, a} \left( \pi(\theta_{-i} \mid \theta_i) \prod_{j \in \Ncal} x_{j,t}(a_j \mid \theta_j) \right) \cdot u_i(a_i', a_{-i} \mid \theta_i) \cdot \left(  Q_i^{\theta_i}(a_i', a_i) + q_i^{\theta_i}(a_i') \right),  
        \end{aligned}
    \end{equation*}
    where we moved the order of summation and regrouped the player strategies $x_{j,t}$ to one product.

    When we average over time, we can identify $\frac{1}{T} \sum_{t = 1}^{T} \pi(\theta_{-i} \mid \theta_i) \prod_{j \in \Ncal} x_{j,t}(a_j \mid \theta_j) 
$ with $\sigma_T(a, \theta_{-i} \mid \theta_i)$:
    \begin{equation*}
        \begin{aligned} 
            \frac{1}{T} \sum_{t = 1}^{T}\left\langle \nabla_i u_i(x_t), \nabla_i h(x_t), \right\rangle 
            &= \sum_{a_i', a, \theta} \sigma_T(a, \theta_{-i} \mid \theta_i) \cdot u_i(a_i', a_{-i} \mid \theta_i) \cdot \left(  Q_i^{\theta_i}(a_i', a_i) + q_i^{\theta_i}(a_i') \right) \\
            &= \sum_{a_i',\theta_i} \mathbb{E}_{(a, \theta) \sim \sigma_T} \left[ u_i(a_i', a_{-i}  \mid \theta_i) \cdot (Q_i^{\theta_i}(a_i', a_i) + q_i^{\theta_i}(a_i')) \mid \theta_i \right].
        \end{aligned}
    \end{equation*}

    By summing over all players, this establishes the equivalence, so we get
    \begin{equation*}
    \limsup_{T \to \infty}
        \sum_{i \in \Ncal} \sum_{a_i', \theta_i} \mathbb{E}_{(a, \theta) \sim \sigma_T} \left[ u_i(a_i', a_{-i}  \mid \theta_i) \cdot (Q_i^{\theta_i}(a_i', a_i) + q_i^{\theta_i}(a_i') \mid \theta_i \right] \le 0.
    \end{equation*}
    Note that this holds for any choice of matrices $Q_i^{\theta_i}$ and vectors $q_i^{\theta_i}$ ($i \in \Ncal$, $\theta_i \in \Theta_i$) that satisfy symmetry, tangency, and non-negativity. We can thus set all of them to zero aside for one player $i$ and one of its types $\theta_i$. This implies 
    \begin{equation*}
        \limsup_{T \to \infty}
        \mathbb{E}_{(a, \theta) \sim \sigma_T} \left[ u_i(a_i', a_{-i}  \mid \theta_i) \cdot (Q_i^{\theta_i}(a_i', a_i) + q_i^{\theta_i}(a_i')\mid \theta_i \right] \leq 0, \qquad \forall i \in \Ncal, \theta_i \in \Theta_i,
    \end{equation*}
    which in turn concludes the proof.

\end{proof}
\section{Numerical experiments}\label{app:numerics}

\subsection{Hardware Information}
\label{app:hardware_information}

We ran our experiments on a small-scale server with an Intel(R) Xeon(R) Silver 4116 CPU with 48 cores (2.10 GHz) and 192 GiB of RAM.
Our experiments are based on Python 3.12 and Gurobi 12.1. We left the CPU usage to Gurobi. Each experiment took between several seconds and several minutes (see "Complexity" in \cref{sec:convergence_analysis}).

\subsection{Wasserstein Distances in the Incomplete-Information game}
The convergence curves for the squared Wasserstein-2 distance exhibit a periodic "zig-zag" pattern for the uniform prior (Figure 3a).
This is an artifact of the uniform discretization grid aligning perfectly with the canonical continuous BNE actions ($1/2 \theta_i$). 
Thus, the oscillation is a predictable effect rather than an algorithmic or numerical instability.

\begin{figure}[!htb]
    \begin{minipage}{0.48\linewidth}
        \centering $\mathrm{Beta}(\alpha,1)$
    \end{minipage}\hfill
    \begin{minipage}{0.48\linewidth}
        \centering $\mathrm{Beta}(\alpha, \alpha)$
    \end{minipage}
    
    \medskip
    \centering

    \begin{subfigure}{0.48\linewidth}
        \centering
        \includegraphics[width=\linewidth]{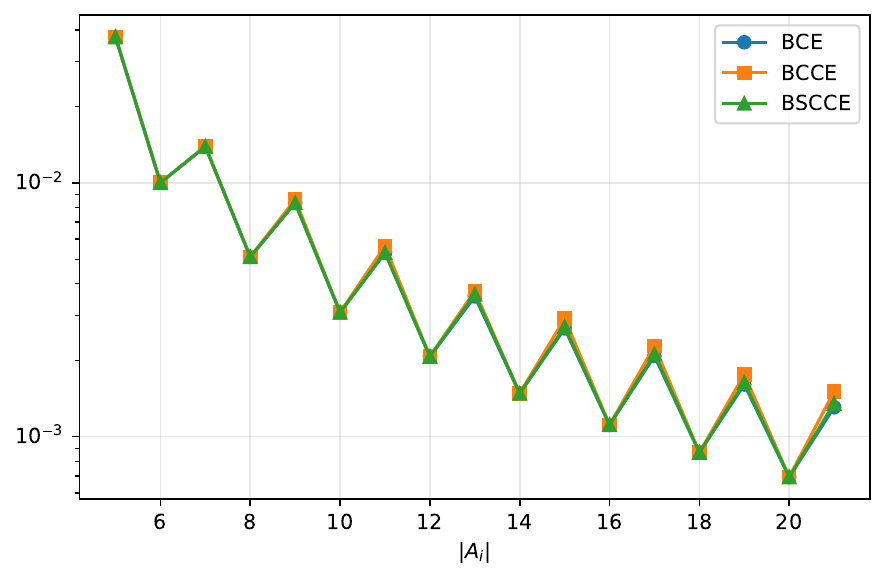}
        \caption{$\mathrm{Beta}(1,1)$ = uniform}
        \label{fig:beta-1-1-uniform}
    \end{subfigure}
    \begin{subfigure}{0.48\linewidth}
        \centering
        \includegraphics[width=\linewidth]{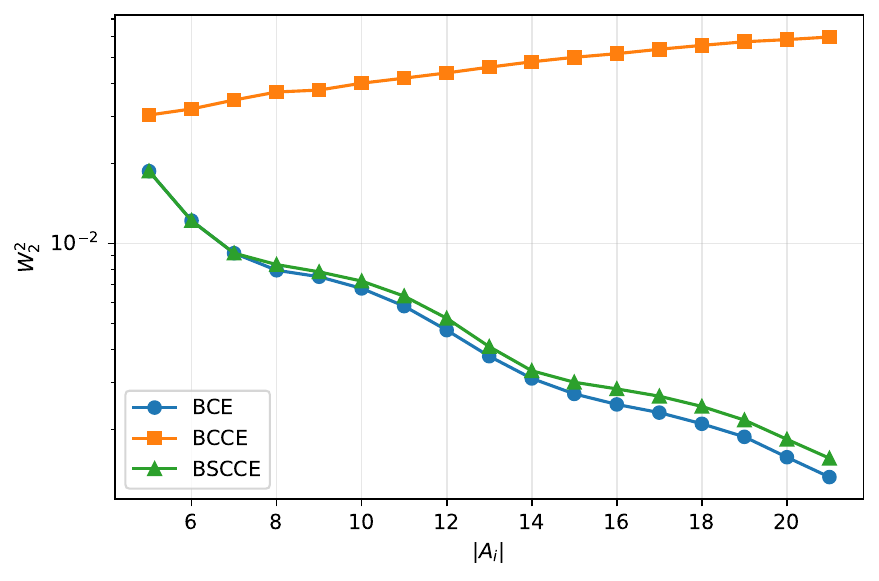}
        \caption{$\mathrm{Beta}(2,2)$}
        \label{fig:beta21-uniform}
    \end{subfigure}

    \begin{subfigure}{0.48\linewidth}
        \centering
        \includegraphics[width=\linewidth]{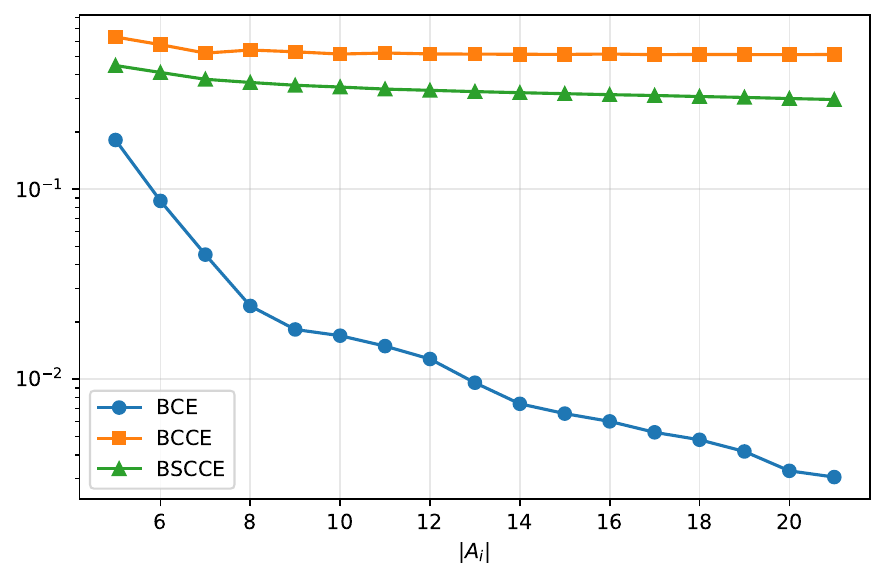}
        \caption{$\mathrm{Beta}(5,1)$}
        \label{fig:beta-5-1-uniform}
    \end{subfigure}
    \begin{subfigure}{0.48\linewidth}
        \centering
        \includegraphics[width=\linewidth]{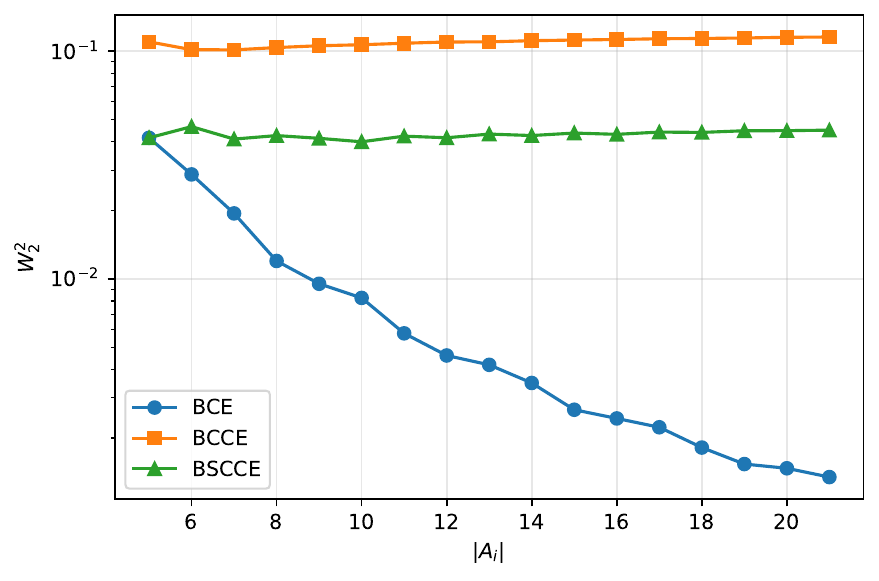}
        \caption{$\mathrm{Beta}(5,5)$}
        \label{fig:beta-2-2-uniform}
    \end{subfigure}

    \begin{subfigure}{0.48\linewidth}
        \centering
        \includegraphics[width=\linewidth]{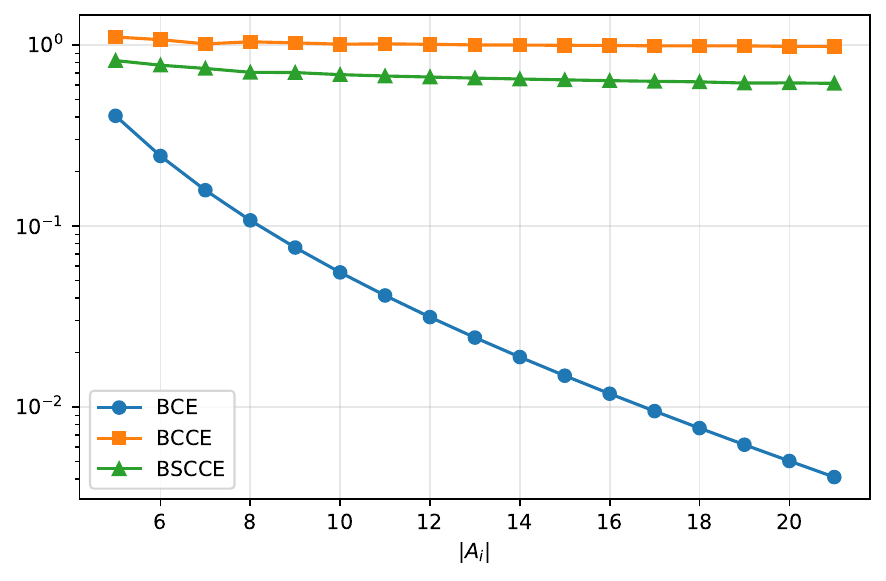}
        \caption{$\mathrm{Beta}(20,1)$}
        \label{fig:beta-20-1-uniform}
    \end{subfigure}
    \begin{subfigure}{0.48\linewidth}
        \centering
        \includegraphics[width=\linewidth]{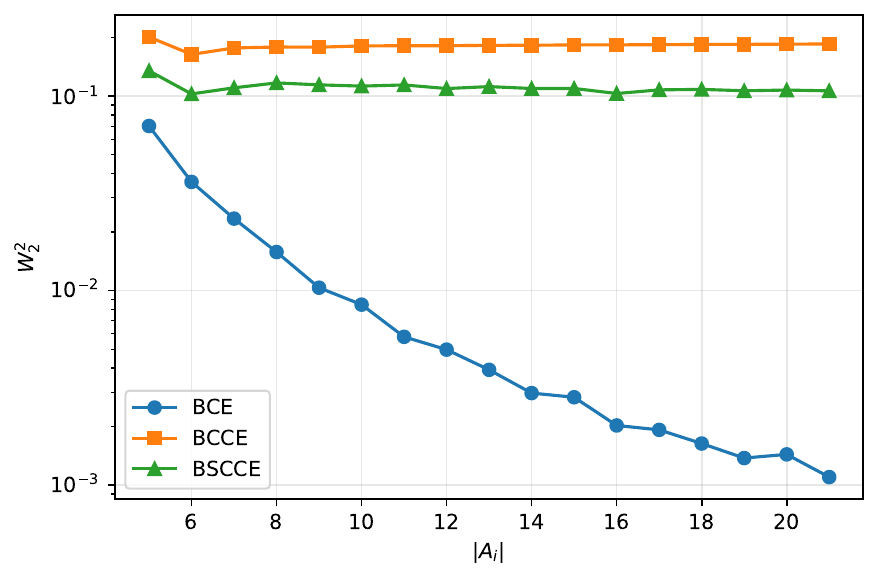}
        \caption{$\mathrm{Beta}(20,20)$}
        \label{fig:beta-20-20-uniform}
    \end{subfigure}
    
    \caption{Dependence of the distance to the canonical BNE on the fineness of uniform discretization for various $\mathrm{Beta}$ priors.}
\end{figure}

\vfill
\pagebreak
\subsection{Equilibrium Polytopes in the Complete-Information Case}

As the previous results suggest that the tightness of the BSCCE concept decreases as priors become degenerate, we investigate the limiting complete-information case.
\Cref{fig:polytopes-CI} illustrates the different equilibrium polytopes for a first-price auction with two players, who both have type $\theta_i = 1$.
In this case, the only possible CEs are that both players perfectly coordinate to play the exact same action out of the top three, whereas the SCCE and CCE sets are much larger.
This implies that, although the SCCE set is smaller than the CCE set, an even stronger equilibrium refinement might be necessary.

\begin{figure}[!htb]
    \centering
    \begin{subfigure}[b]{0.32\linewidth}
        \centering
        \includegraphics[width=\linewidth]{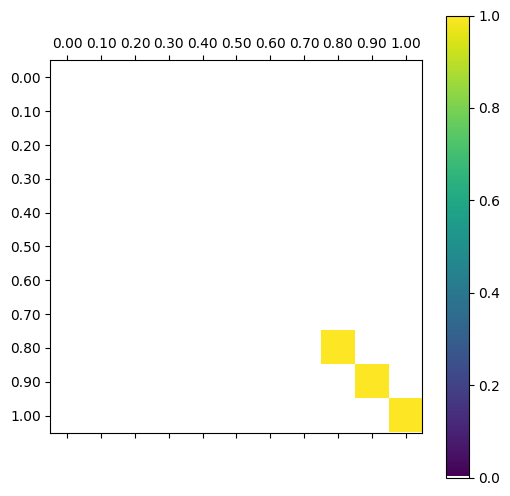}
        \caption{CE polytope}
        \label{fig:ce_polytope-CI}
    \end{subfigure}
    \hfill
    \begin{subfigure}[b]{0.32\linewidth}
        \centering
        \includegraphics[width=\linewidth]{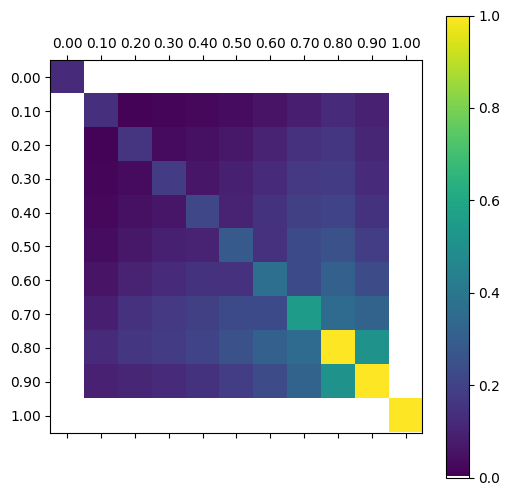}
        \caption{SCCE polytope}
        \label{fig:scce_polytope-CI}
    \end{subfigure}
    \hfill
    \begin{subfigure}[b]{0.32\linewidth}
        \centering
        \includegraphics[width=\linewidth]{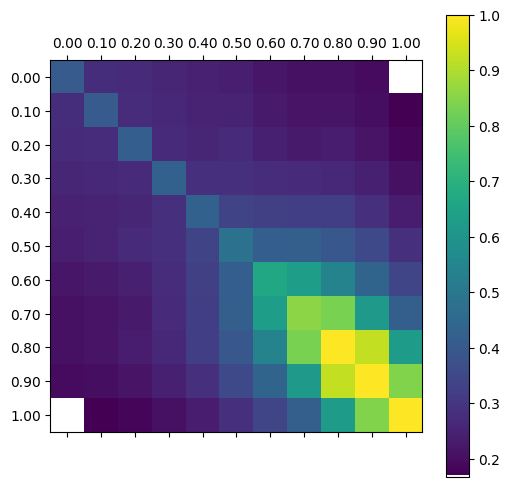}
        \caption{CCE polytope}
        \label{fig:cce_polytope-CI}
    \end{subfigure}
    \caption{Comparison of equilibrium polytopes for a first price auction where both players have type $\theta_i=1$.}
    \label{fig:polytopes-CI}
\end{figure}


\end{document}